\documentclass[a4paper,USenglish]{article}
\usepackage{makeidx}  
\usepackage{algorithm}
\usepackage{algorithmicx}
\usepackage[noend]{algpseudocode}
\usepackage{amsmath,amsfonts,mathrsfs,amsthm,amssymb}
\usepackage{xspace}
\usepackage{graphicx}
\usepackage{epsfig}
\usepackage{hyperref}
\usepackage[margin=1.5in]{geometry}

\newtheorem{theorem}{Theorem}
\newtheorem{lemma}[theorem]{Lemma}
\newtheorem{claim}[theorem]{Claim}

\newtheorem{corollary}[theorem]{Corollary}
\newtheorem{definition}[theorem]{Definition}

\newcommand{\sol}{\ensuremath{\mathsf{\tiny M'}}\xspace}
\newcommand{\solI}{\ensuremath{\mathsf{\tiny M}}\xspace}
\newcommand{\solb}{\ensuremath{\mathsf{\tiny B}}\xspace}
\newcommand{\sola}{\ensuremath{\mathsf{\tiny A}}\xspace}

\newcommand{\ratio}{\frac{25}{3}}
\newcommand{\iratio}{\frac{3}{25}}

\newcommand{\sse}{\subseteq}

\newcommand{\OPT}{\mbox{\sc OPT}}
\newcommand{\lp}{{\small \textsf{LP}}\xspace}

\algnewcommand\algorithmicinput{\textbf{Input:}}
\algnewcommand\INPUT{\item[\algorithmicinput]}
\algnewcommand\algorithmicoutput{\textbf{Output:}}
\algnewcommand\OUTPUT{\item[\algorithmicoutput]}

\newcommand{\there}{{\texttt{here}}\xspace}
\newcommand{\tup}{{\texttt{up}}\xspace}
\newcommand{\tleft}{{\texttt{left}}\xspace}
\newcommand{\tright}{{\texttt{right}}\xspace}

\newcommand{\bm}{{\textsc{$b$-Matching}}\xspace}
\newcommand{\dm}{{\textsc{Demand Matching}}\xspace}
\newcommand{\dmm}{{\textsc{Matroidal Demand Matching}}\xspace}
\newcommand{\gdm}{{\textsc{Generalized Demand Matching}}\xspace}
\newcommand{\cp}{{\textsc{Coupled Placement}}\xspace}

\newcommand{\gdms}{\ensuremath{\mathsf{GDM}}\xspace}
\newcommand{\dms}{\ensuremath{\mathsf{DM}}\xspace}
\newcommand{\gdmm}{\ensuremath{\mathsf{GDM_M}}\xspace}

\newcommand{\cM}{{\mathcal M}\xspace}

\title{Further Approximations for Demand Matching: Matroid Constraints and Minor-Closed Graphs}

\author{
Sara Ahmadian \thanks{Department of Combinatorics and Optimization, University of Waterloo.}
\and
Zachary Friggstad \thanks{Department of Computing Science, University of Alberta. This research was undertaken, in part, thanks to funding from the Canada Research Chairs program and an NSERC Discovery Grant.}
}

\begin{document}
\maketitle

\begin{abstract}
We pursue a study of the \gdm problem, a common generalization of the \bm and \textsc{Knapsack} problems. Here, we are given a graph with vertex capacities, edge profits, and asymmetric demands on the edges.
The goal is to find a maximum-profit subset of edges so the demands of chosen edges do not violate the vertex capacities.
This problem is {\bf APX}-hard and constant-factor approximations are already known.

Our main results fall into two categories. First, using iterated relaxation and various filtering strategies,
we show with an efficient rounding algorithm that if an additional matroid structure $\cM$ is given and we further only allow sets $F \subseteq E$ that are independent in $\cM$, the natural LP relaxation
has an integrality gap of at most $\ratio \approx 8.333$. This can be further improved in various special cases, for example we improve over the 15-approximation for the previously-studied \cp problem [Korupolu et al. 2014] by giving a $7$-approximation.

Using similar techniques, we show the problem of computing a minimum-cost base in $\mathcal M$ satisfying vertex capacities admits a $(1,3)$-bicriteria approximation: the cost is at most the optimum and the
capacities are violated by a factor of at most 3. This improves over the previous $(1,4)$-approximation in the
special case that $\cM$ is the graphic matroid over the given graph [Fukanaga and Nagamochi, 2009].

Second, we show \dm admits a polynomial-time approximation scheme in graphs that exclude a fixed minor.
If all demands are polynomially-bounded integers, 
this is somewhat easy using dynamic programming in bounded-treewidth graphs.
Our main technical contribution is a sparsification lemma that allows us to scale the demands of some items to be used in a more intricate dynamic programming
algorithm, followed by some randomized rounding to filter our scaled-demand solution to one whose original demands satisfy all constraints.
\end{abstract}



\section{Introduction} \label{sec:intro}

Many difficult combinatorial optimization problems involve resource allocation. Typically, we have a collection of resources, each with finite supply or {\bf capacity}. Additionally there are tasks to be accomplished,
each with certain requirements or {\bf demands} for various resources. Frequently the goal is to select a maximum value set of tasks and allocate the required amount of resources to each task while
ensuring we have enough resources to accomplish the chosen tasks. This is a very well-studied paradigm: classic problems include \textsc{Knapsack}, \textsc{Maximum Matching}, and \textsc{Maximum Independent Set},
and more recently-studied problems include \textsc{Unsplittable Flow} \cite{AGLW12} and \textsc{Coupled Placement} \cite{KMRT14}. In general, we cannot hope to get non-trivial approximation algorithms for these problems.
Even the simple setting of \textsc{Maximum Independent Set} is inapproximable \cite{H99,Z07}, so research frequently focuses on well-structured special cases.

Our primary focus is when each task requires at most two different resources. Formally, in \gdm (\gdms)
we are given a graph $G = (V,E)$ with, perhaps, parallel edges. The vertices should be thought of as resources and the tasks as edges.
Each $v \in V$ has a capacity $b_v \geq 0$ and each $uv \in E$ has demands $d_{u,e}, d_{v,e} \geq 0$ and a value $p_{uv} \geq 0$. A subset $M \subseteq E$ is {\em feasible} if $d_v(\delta(v) \cap M) \leq b_v$
for each $v \in V$ (we use $d_v(S)$ as shorthand for $\sum_{e \in S} d_{v, e}$ when $S \subseteq \delta(v)$).
We note that the simpler term \dm (\dms) is used when $d_{u,e} = d_{v,e}$ for each edge $e = uv$ (e.g. \cite{SV07,SW16}).

\dms is well-studied from the perspective of approximation algorithms. It is fairly easy to get constant-factor approximations and some work has been done refining these constants.
Moreover the integrality gap of a natural LP relaxation is also known to be no worse than a constant
(see the related work section). On the other hand, \dms is {\bf APX}-hard \cite{SV07}.

Our main results come in two flavours. First, we look to a generalization we call \dmm (\gdmm). Here, we are given the same input as in \gdms but there is also a matroid $\cM = (E,\mathcal I)$ over the edges $E$ with independence
system $\mathcal I \subseteq 2^E$
that further restricts feasibility of a solution. A set $F \subseteq E$ is feasible if it is feasible as a solution to the underlying \gdms problem and also $F \in \mathcal I$.
We assume $\cM$ is given by an efficient independence oracle. Our algorithms will run in time that is polynomial in the size of $G$ and the maximum running time of the independence oracle.

As a special case, \gdmm includes the previously-studied \textsc{Coupled Placement} problem. In \textsc{Coupled Placement}, we are given a bipartite graph $G = (V,E)$ with vertex capacities. The tasks
are not individual edges, rather for each task $j$ and each $e=uv \in E$ we have demands $d^j_{u,e}, d^j_{v,e}$ placed on the respective endpoints $u,v$ for placing $j$ on edge $e$.
Finally, each task $j$ has a profit $p_j$ and the goal is to select a maximum-profit subset of tasks $j$
and, for each chosen task $j$, assign $j$ to an edge of $G$ so vertex capacities are not violated. We note that an edge may receive many different tasks. This can be viewed as an instance of \gdmm by creating parallel copies
of each edge $e \in E$, one for each task $j$ with corresponding demand values and profit for $j$ and letting $\cM$ be the partition matroid ensuring we take at most one edge corresponding to any task.

For another interesting case, consider an instance where, in addition to tasks requiring resources from a shared pool, each also needs to be connected to a nearby power outlet.
We can model such an instance by letting $\mathcal M$ be a transversal matroid over a bipartite graph where tasks form one side, outlets form the other side, and an edge indicates the edge can reach the outlet.

In fact \gdmm can be viewed as a packing problem with a particular submodular objective function. These are studied in \cite{BKNS12} so the problem is not new; our results are improved approximations.
Our techniques also apply to give bicriteria approximations for the variant of \gdmm where we must pack a cheap base of the matroid while obeying congestion bounds. In the special case where $\cM$
is the graphic matroid over $G$ itself (i.e. the \textsc{Minimum Bounded-Congestion Spanning Tree} problem), we get an improved bicriteria approximation.

Second, we study \gdms in special graph classes. In particular, we demonstrate a PTAS in families of graphs that exclude a fixed minor.
This is complemented by showing that even \dms is strongly NP-hard in simple planar graphs, thereby ruling out a fully-polynomial time approximation scheme (FPTAS) in simple planar graphs unless {\bf P} = {\bf NP}.


\subsection{Statements of Results and Techniques}\label{sec:results}

We first establish some notation. For a matroid $\mathcal M = (E, \mathcal I)$, we let $r_\cM : 2^E \rightarrow \mathbb Z_{\geq 0}$ be the rank function for $\mathcal M$. We omit the subscript $\cM$
if the matroid is clear from the context. For $v \in V$
we let $\delta(v)$ be all edges having $v$ as one endpoint; for $F \subseteq E$ we let $\delta_F(v)$ denote $\delta(v) \cap F$. For a vector of values $x$ indexed by a set $S$,
we let $x(A) = \sum_{i \in A} x_i$ for any $A \subseteq S$. A {\bf polynomial-time approximation scheme} (PTAS) is an approximation algorithm that accepts an additional parameter $\epsilon > 0$.
It finds a $(1+\epsilon)$-approximation in time $O(n^{f(\epsilon)})$ for some function $f$ (where $n$ is the size of the input apart from $\epsilon$),
so the running time is polynomial for any constant $\epsilon > 0$. An FPTAS is a PTAS with running time being polynomial in $\frac{1}{\epsilon}$ and $n$.

We say an instance of \gdms has a {\bf consistent ordering of edges} if $E$ can be ordered such that the restriction of this ordering to each set $\delta(v)$ has
these edges $e \in E$ appear in nondecreasing order of demands $d_{v,e}$. For example, \dms itself has a consistent ordering of demands, just sort edges by their demand values.
This more general case was studied in \cite{O11}. We say the instance is {\bf conflict-free} if for any $e,f \in E$ we have that $\{e,f\}$ does not violate the capacity of any vertex.

In the first half of our paper, we mostly study the following linear-programming relaxation of \gdmm. Here, $r : 2^E \rightarrow \mathbb Z$ is the rank function for $\mathcal M$.
\begin{equation}\tag{{\bf LP-M}} \label{lp-m}
{\rm max}: \left\{\sum_{e \in E} p_e x_e : \sum_{e \in \delta(v)} d_{v,e} x_e \leq b_v~ \forall v \in V, ~~ x(A) \leq r(A)~~ \forall A \subseteq E,~~ x \geq 0 \right\}
\end{equation}
Note $x(\{e\}) \leq 1$ is enforced for each $e \in E$ as $r(\{e\}) \leq 1$.
It is well-known that the constraints can be separated in polynomial time when given an efficient independence oracle for $\mathcal M$,
so we can find an extreme point optimum solution to \eqref{lp-m} in polynomial time.


Throughout, we assume each edge is feasible by itself. This is without loss of generality: an edge that is infeasible by itself can be discarded\footnote{This is a standard step when studying packing LPs, even the natural
\textsc{Knapsack} LP relaxation has an unbounded integrality gap if some items do not fit by themselves.}.
We first prove the following.
\begin{theorem}\label{thm:main}
Let $\OPT\eqref{lp-m}$ denote the optimum solution value of \eqref{lp-m}. If $d_{v,e} \leq b_v$ for each $v \in V, e \in \delta(v)$ then we can find, in polynomial time,
a feasible solution $\solI \subseteq E$ such that $\OPT\eqref{lp-m} / p(\solI)$ (and, thus, the integrality gap) is at most:
\begin{itemize}
\item $\ratio$ in general graphs
\item $7$ in bipartite graphs
\item $5$ if the instance has a consistent ordering of edges
\item $4$ if the instance is conflict-free
\item $1+O(\epsilon^{1/3})$ if $d_{v,e} \leq \epsilon\cdot b_v$ for each $v \in V, e \in \delta(v)$ (i.e. edges are {\bf $\epsilon$-small})
\end{itemize}
\end{theorem}
These bounds also apply to graphs with parallel edges, so we get a $7$-approximation for \textsc{Coupled Placement},
which beats the previously-stated $15$-approximation in \cite{KMRT14}.

We prove all bounds in Theorem \ref{thm:main} using the same framework: iterated relaxation to find some $\sol \in \mathcal I$ with $p(\sol) \geq OPT_{LP}$ that may violate some capacities by a
controlled amount, followed by various strategies to pare the solution down to a feasible solution.
We note constant-factor approximations for \gdmm were already implicit in \cite{BKNS12}, the bounds in Theorem \ref{thm:main} improve over their bounds
and are relative to \eqref{lp-m} whereas \cite{BKNS12} involves multilinear extensions of submodular functions.


Our techniques can also be used to address a variant of \gdmm. The input is the same, except we are required to select a base of $\mathcal M$. The goal is to find a minimum-value base
satisfying the vertex capacities.
More formally, let \textsc{Minimum Bounded-Congestion Matroid Basis} be given the same way as in \gdmm, except the goal is to find a minimum-cost base \solb of $\cM$
satisfying the vertex capacities (i.e. the cheapest base that is a solution to the \gdmm problem).

When all demands are 1, this is the \textsc{Minimum Bounded-Degree Matroid Basis} problem which, itself, contains the famous \textsc{Minimum Bounded-Degree Spanning Tree} problem.
As an important special case, we let \textsc{Minimum Bounded-Congestion Spanning Tree} denote the problem when $k=2$ with arbitrary demands where $\cM$ is the graphic matroid over $G$.
Even determining if there is a feasible solution is {\bf NP}-hard, so we settle with approximations that may violate the capacities a bit.
Consider the following LP relaxation, which we write when $G$ can even be a hypergraph.
\begin{equation}\tag{{\bf LP-B}} \label{lp-b}
{\rm min}: \left\{\sum_{e \in E} p_e x_e : \sum_{e \in \delta(v)} d_{v,e} x_e \leq b_v~ \forall v \in V, ~~x(A) \leq r(A)~ \forall A \subseteq E,~~ x(E) = r(E), ~~ x \geq 0 \right\}
\end{equation}
As a side effect of how we prove Theorem \ref{thm:main}, we also prove the following.
\begin{corollary}\label{cor:congestion}
If $G$ is a hypergraph where each edge has size at most $k$, then in polynomial time we can either determine there is no integral point in \eqref{lp-b}
or we can find a base $\solb$ of $\mathcal M$ such that $p(\solb) \leq \OPT\eqref{lp-b}$
and $d_v(\delta_\solb(v)) \leq b_u + k \cdot \max_{e \in \delta(v)} d_{v,e}$ for each $v \in V$.
\end{corollary}
\begin{theorem}\label{thm:bic}
There is a $(1, 1+k)$-bicriteria approximation for \textsc{Minimum Bounded-Congestion Matroid Basis}.
\end{theorem}
In particular, there is a $(1,3)$-bicriteria approximation for \textsc{Minimum Bounded-Congestion Spanning Tree}, beating the previous best $(1,4)$-bicriteria approximation \cite{FN09}.
Theorem \ref{thm:bic} matches the bound in \cite{KMRT14} for the special case of \textsc{Coupled Placement} in $k$-partite hypergraphs, but in a more general setting.

One could also ask if we can generalize Theorem \ref{thm:main} to hypergraphs. An $O(k)$-approximation is already known \cite{BKNS12} and the integrality gap
of \eqref{lp-m} is $\Omega(k)$ even without matroid constraints,
so we could not hope for an asymptotically better approximation. We remind the reader that our focus in \gdmm is improved constants in the case of graphs ($k=2$).

Our second class of results are quite easy to state. We study \gdms in families of graphs that exclude a fixed minor. It is easy to see \gdms is strongly {\bf NP}-hard in planar graphs if one allows parallel edges as it is even strongly
{\bf NP}-hard with just two vertices, {\em e.g.} see \cite{GL80,MC84}. We show the presence of parallel edges is not the only obstacle to getting an FPTAS for \gdms (or even \dms) in planar graphs.
\begin{theorem}\label{thm:planar_hard}
\dms is NP-hard in simple, bipartite planar graphs even if all demands, capacities, values, and vertex degrees are integers bounded by a constant.
\end{theorem}

We then present our main result in this vein, which gives a PTAS for \gdms in planar graphs among other graph classes.
\begin{theorem}\label{thm:ptas}
\gdms admits a PTAS in families of graphs that exclude a fixed minor.
\end{theorem}
This is obtained through the usual reduction to bounded-treewidth graphs \cite{EHK05}.
We would like to scale demands to be polynomially-bounded integers, as then it is easy to solve the problem using dynamic programming over the tree decomposition. But packing problems
are too fragile for scaling demands naively: an infeasible solution may be regarded as feasible in the scaled instance.

We circumvent this issue with a sparsification lemma showing there is a near-optimal solution $\sol$ where, for each vertex $v$, after packing a constant number of edges across $v$
the remaining edges in $\delta_\sol(v)$ have very small demand compared with even the residual capacity. Our dynamic programming algorithm then guesses these large edges in each bag
of the tree decomposition and packs the remaining edges according to scaled values. The resulting solution may be slightly infeasible, but the blame rests on our scaling of {\em small} edges
and certain pruning techniques can be used to whittle this solution down to a feasible solution with little loss in the profit.


\subsection{Related Work}
\dms (the case with symmetric demands) is well-studied. Shepherd and Vetta initially give a 3.264-approximation in general graphs and a 2.764-approximation in bipartite graphs \cite{SV07}. These are all with
respect to the natural LP relaxation, namely \eqref{lp-m} with matroid constraints replaced by $x_e \leq 1, \forall e \in E$. They also prove that \dms
is {\bf APX}-hard even in bipartite graphs and give an FPTAS in the case $G$ is a tree.

Parekh~\cite{O11} improved the integrality gap bound for general graphs to 3 in cases of \gdms that have a consistent ordering of edges. Singh and Wu improve the gap in bipartite graphs
to 2.709 \cite{SW16}. The lower bound on the integrality gap for general graphs is 3 \cite{SV07}, so the bound in \cite{O11} is tight. In bipartite graphs, the gap is at least 2.699 \cite{SW16}.

Bansal, Korula, Nagarajan, and Srinivasan study the generalization of \gdms to hypergraphs \cite{BKNS12}. They show if each edge has at most $k$ endpoints, the integrality gap of the natural LP relaxation is
$\Theta(k)$. They also prove that a slight strengthening of this LP has a gap of at most $(e+o(1)) \cdot k$. Even more relevant to our results is that they prove if the value function over the edges is
submodular, then rounding a relaxation based on the multilinear extension of submodular functions yields a $\left(\frac{e^2}{e-1} + o(1) \right) \cdot k$-approximation. For $k=2$, this immediately gives a constant-factor
approximation for \gdmm by considering the submodular objective function $f : 2^E \rightarrow \mathbb R$ given by $f(S) = \max\{ p(S') : S' \subseteq S, S' \in \mathcal I\}$.

They briefly comment on the case $k=2$
in their work and say that even optimizations to their analysis for this special case yields only a 11.6-approximation for \dms (i.e. without a matroid constraint).
So our $\ratio$-approximation for \gdmm is an improvement over their work.
They also study the case where $d_{v,e} \leq \epsilon \cdot b_v$ for each $v \in V$ and each hyperedge $e \in \delta(v)$
and present an algorithm for \gdms with submodular objective functions whose approximation guarantee tends to $\frac{4e^2}{e-1}$ as $\epsilon \rightarrow 0$ (with $k$ fixed).

As noted earlier, our results yield improvements for two specific problems.
First, our 7-approximation for \gdmm in bipartite graphs improves over the 15-approximation for \textsc{Coupled Placement} \cite{KMRT14}.
The generalization of \textsc{Coupled Placement} to $k$-partite hypergraphs is also studied in \cite{KMRT14} where they obtain an $O(k^3)$-approximation, but this was already inferior to the $O(k)$-approximation in \cite{BKNS12}
when viewing it as a submodular optimization problem
with packing constraints.

Second, our work also applies to the \textsc{Minimum-Congestion Spanning Tree} problem, defined earlier. 
Determining if there is even a feasible solution is {\bf NP}-hard as this models the Hamiltonian Path problem. A famous result of Singh and Lau
shows if all demands are 1 (so we want to bound the degrees of the vertices) then we can find a spanning tree with cost at most the optimum cost (if there is any solution) that violates the degree bounds additively by +1
\cite{SL07}. In the case of arbitrary demands, the best approximation so far is a $(1,4)$-approximation \cite{FN09}: it finds a spanning tree whose cost is at most the optimal cost and violates the capacities by a factor of at most 4.
It is known that obtaining a $(1,c)$-approximation is {\bf NP}-hard for any $c < 2$ \cite{G07}.



\section{Approximation Algorithms for Generalized Demand Matching over Matroids}\label{sec:alg}

Here we present approximation algorithms for \gdmm and prove Theorem~\ref{thm:main} and Corollary~\ref{cor:congestion}. Our algorithm consists of two phases: the iterative relaxation phase and the pruning phase.
The first finds a set $\sol \in \mathcal I$ with $p(\sol) \geq \OPT\eqref{lp-m}$ that places demand at most $b_v + 2 \cdot \max_{e \in \delta_\sol(v)} d_{v,e}$ on each $v \in V$.
The second prunes $\sol$ to a feasible solution, different pruning strategies are employed to prove the various bounds in Theorem \ref{thm:main}.

\subsection{Iterative Relaxation Phase}
This part is presented for the more general case of hypergraphs where each edge has at most $k$ endpoints. Our \gdmm results in Theorem \ref{thm:main} pertain to $k=2$, but we will use properties
of this phase in our proof of Corollary \ref{cor:congestion}.
The algorithm starts with \eqref{lp-m} and iteratively removes edge variables and vertex capacities.


We use the following notation. For some $W \subseteq V, F \subseteq E$, a matroid $\mathcal M'$ with ground set $F$, and values $b'_v, v \in W$ we let \ref{lp-m}$[W,F,\mathcal M',b']$
denote the LP relaxation we get from \eqref{lp-m} over the graph $(V,F)$ with matroid $\cM'$ where we drop capacity constraints for $v \in V-W$ and use capacities $b'_v$ for $v \in W$.

Note that the relevant graph for \ref{lp-m}$[W,F,\mathcal M',b']$ still has all vertices $V$, it is just that some of the capacity constraints are dropped.
Also, for a matroid $\cM'$ and an edge $e \in F$ we let $\cM'- e$ be the matroid obtained by deleting $e$ and, if $\{e\}$ is independent in $\cM'$,
we let $\cM' / e$ be the matroid obtained by contracting $e$ (i.e. a set $A$ is independent in $\cM' / e$ if and only if $A \cup \{e\}$ is independent in $\cM'$).


\begin{algorithm}
\caption{Iterated Relaxation Procedure for \gdmm} \label{alg:iter}
\begin{algorithmic}
\State $W \leftarrow V, F \leftarrow E, \cM' \leftarrow \cM$
\State $b'_v \leftarrow b_v$ for each $v \in V$
\State $\sol \leftarrow \emptyset$
\While{$F \neq \emptyset$}
\State solve \ref{lp-m}$[W,F,\cM',b']$ to get an optimum extreme point $x^*$
\If{$x^*_e = 0$ for some $e \in F$}
\State $F \leftarrow F-\{e\}$
\State $\cM' \leftarrow \cM' - e$ \Comment{fix $x^*_e$ to 0 from now on}
\ElsIf{$x^*_e = 1$ for some $e \in F$}
\State $F \leftarrow F-\{e\}$
\State $\cM' \leftarrow \cM' / e$
\State $\sol \leftarrow \sol \cup \{e\}$ \Comment{fix $x^*_e$ to 1 from now on}
\State $b'_v \leftarrow b'_v-d_{v,e}$ for each endpoint $v$ of $e$ \Comment{permanently allocate space for $e$}
\Else
\State let $v$ be any vertex in $W$ with minimum value $|\delta_F(v)| - x^*(\delta_F(v))$
\State $W \leftarrow W - \{v\}$ \Comment{drop the capacity constraint for $v$}
\EndIf
\EndWhile
\State \Return $\sol$
\end{algorithmic}
\end{algorithm}

Algorithm \ref{alg:iter} describes the steps in the iterated relaxation phase.
Correct execution and termination are consequences of the following two lemmas. Their proofs are standard for iterated techniques.
\begin{lemma}\label{lem:matroid}
Throughout the execution of the algorithm, whenever $\cM'$ is contracted by $e$ we have $\{e\}$ is independent (i.e. $e$ is not a loop) in $\cM'$.
\end{lemma}
\begin{proof}
This is simply because $x^*$ is a feasible solution to \ref{lp-m}$[W,F,\cM',b']$, so whenever $\cM'$ is contracted by $e$ we have $1 = x^*_e \leq r_{\cM'}(\{e\})$.
That is, $\{e\}$ is independent in $\cM'$.
\end{proof}

\begin{lemma}\label{lem:terminate}
The algorithm terminates in polynomial time and the returned set \sol is an independent set in $\cM$ with $p(\sol) \geq \OPT\eqref{lp-m}$.
Furthermore, if at any point $W' = \emptyset$ then the corresponding extreme point solution $x^*$ is integral.
\end{lemma}
\begin{proof}
Each iteration can be executed in polynomial time. The only thing to comment on here is that \ref{lp-m}$[W,F,\cM',b']$ can be solved in polynomial time
because we assume $\cM$ is given by an efficient separation oracle (so we also get one for each $\cM'$ encountered in the algorithm), and this suffices to separate the
constraints (e.g. Corollary 40.4a in \cite{schrijver}).

Next we consider termination.
Note that optimal solution $x^*$ in one step induces a a feasible solution for the \lp considered in
the next step by ignoring the edge that was discarded or fixed in this iteration (if any).
As the initial LP is feasible (e.g. using $x_e = 0$ for all $e \in E$), the LP remains feasible. 
Each iteration removes an edge from $F$ or a vertex from $W$. If $W$ ever becomes empty, then the only constraints defining
\ref{lp-m}$[W,F,\cM',b']$ are the matroid rank constraints. It is well-known such polytopes are integral (e.g. Corollary 40.2b in \cite{schrijver}), so the algorithm will remove an edge
in every subsequent iteration. That is, the algorithm terminates within $|E| + |V|$ iterations.

Finally, to bound $p(\sol)$ note that if an edge is dropped or a vertex constraint is relaxed in an iteration, the optimum solution value of the resulting LP does not decrease.
If an edge $e$ is added to $\sol$, the optimum solution of the value drops by at most $p_e$ since the restriction of $x^*$ to $F-\{e\}$ remains feasible and $p(\sol)$ increases by exactly $p_e$.
So, inductively, we have the returned set $\cM'$ satisfying $p(\cM') \geq \OPT\eqref{lp-m}$.
\end{proof}

The last statement in Lemma \ref{lem:terminate} emphasizes the last case in the body of the loop cannot be encountered if $W' = \emptyset$.

Next, we prove $\sol$ is a feasible demand matching with respect to capacities $b_v + k\cdot \max_{e\in \delta(v)} d_{e,v}$ for each $v\in V$ by utilizing the following claim. 
 
\begin{claim}\label{clm:atmosttwo}
In any iteration, if $0 < x^*_e < 1$ for each $e \in F$ then $|\delta_F(v)| \leq x^*(\delta_F(v)) + k$ for some $v \in W$.
\end{claim}
\begin{proof}
Let $A_1 \subsetneq A_2 \subsetneq \ldots \subsetneq A_t \subseteq F$ be any {\em maximal-length chain of tight sets}. That is, $x^*(A_i) = r_{\cM'}(A_i)$ for each $A_i$ in the chain.
Then the indicator vectors $\chi_{A_i} \in \{0,1\}^F$ of the sets $A_i$ are linearly independent and every other $A \subseteq F$
with $x^*(A_i) = r_{\cM'}(A)$ has $\chi_A \in {\rm span}\{\chi_{A_i} : 1 \leq i \leq t\}$. This can be proven by using uncrossing techniques that exploit submodularity of $r_{\cM'}$, see Chapter $5$ of \cite{lau2011iterative}.

Now, as $A_{i-1} \subsetneq A_i$ for $1 < i \leq t$ and $x^*_e > 0$ for each $e \in F$, we see $r_{\cM'}(A_i) = x^*(A_i) > x^*(A_{i-1}) = r_{\cM'}(A_{i-1})$.
Since the ranks are integral and $r_{\cM'}(A_1) \neq 0$ (as $A_1 \neq \emptyset$ so $r(A_1) = x^(A_1) > 0$), then $r_{\cM}(A_i) \geq i$ for all $1 \leq i \leq t$.

Note that $|F| \leq t + |W|$ because the number of non-zero (fractional) variables is at most the size of a basis for the tight constraints.
We have
\begin{eqnarray*}
\sum_{v \in W} |\delta_{F}(v)| - x^*(\delta_{F}(v))
 & \leq & \sum_{v \in V} |\delta_{F}(v)| - x^*(\delta_{F}(v)) \leq k \cdot (|F|-x^*(F)) \\
 & \leq & k \cdot (|F|-r_{\cM'}(A_t)) \leq  k \cdot (|F|-t) \leq k \cdot |W|.
 \end{eqnarray*} 
The second bound holds because each edge has at most $k$ endpoints, so it can contribute $1 - x^*_e \geq 0$ at most $k$ times throughout the sum.
Thus, some $v \in W$ satisfies the claim.
\end{proof} 
 
\begin{lemma}\label{lem:3bv}
Algorithm \ref{alg:iter} returns a set $\sol \in \mathcal I$ such that
$d_v(\delta_{\sol}(v) - L(v)) \leq b_v$
where $L(v)$ denotes the $\min\{k, |\delta_{\sol}(v)|\}$ edges $e \in \delta_{\sol}(v)$ with greatest demand $d_{v,e}$ across $v$.
\end{lemma}
\begin{proof}
We know $\sol \in \mathcal I$ by Lemma \ref{lem:terminate}.
Consider an iteration where a vertex $v \in W$ is removed from $W$.
Claim~\ref{clm:atmosttwo} shows $|\delta_{F}(v)| \leq x^*(\delta_{F}(v)) + k$.

Let $F^k_v = \{e_1, \ldots, e_k\}$ be the $k$ edges of this iteration in $\delta_{F}(v)$ having largest demand (if $|\delta_F(v)| < k$ then let $F^k_v = \delta_F(v)$). Then
$$\sum_{e\in \delta_{F}(v) - F^k_v} d_{v,e} \leq \sum_{e\in \delta_{F}(v)} d_{v,e} \cdot x^*_{e,v} \leq b'_v.$$
The first bound follows because if we shift $x^*$-values from larger- to smaller-demand edges the value $\sum_{e\in \delta_{F}(v)} d_{v,e} x^*_{e,v}$ does not increase.
We can continue to do this until each $e \in \delta_F(v) - F^k_v$ has one unit of $x^*$-mass because $|\delta_F(v)| - k \leq x^*(\delta_F(v))$.

At this point of the algorithm, we have $d_v(\delta_{\solI}(v)) = b_v - b'_v$ (letting $\solI$ denote the set $\sol$ from the current iteration).
So $d_v(\delta_F(v) - F^k_v) + d_v(\delta_{\solI}(v)) \leq b_v$.
We conclude by noting the edges returned by the algorithm contains only edges in $\solI \cup F$ so $d_v(\sol - L(v)) \leq b_v$.
\end{proof} 

\begin{proof}[Proof of Corollary \ref{cor:congestion}]
If there is no feasible solution to \eqref{lp-m}, then there can be no integral solution.
Otherwise, we use the same iterated relaxation technique as in Algorithm \ref{alg:iter},
except on \eqref{lp-b}, whose polytope is the restriction of the polytope from \eqref{lp-m} to the base polytope of $\cM$ (which is also integral, Corollary 40.2d of \cite{schrijver}).

All arguments are proven in essentially the same way. So we can find, in polynomial time, a base \solb with $p(\solb) \leq \OPT\eqref{lp-b}$ where $d_v(\delta_{\solb}(v)) \leq b_v + k \cdot \max_{e \in \delta_{\solb}(v)} d_{v,e}$.
\end{proof}

\subsection{Pruning phase}
We focus on \gdmm ($k=2$) in this section and show how to prune a set $\sol \subseteq E$ satisfying the properties of Lemma \ref{lem:3bv} to a feasible solution $\solI \subseteq \sol$ while controlling the loss in its value.
Each part of Theorem~\ref{thm:main} is proved through the following lemmas. In each, for a vertex $v \in V$ we let $L(v)$ be the two edges with highest $d_v$-value in $\delta_\sol(v)$ (or $L(v) = \delta_\sol(v)$
if $|\delta_\sol(v)| \leq 1$). We also let $S(v) = \delta_{\sol}(v) - L(v)$ be the remaining edges. Note $d_v(\delta_{S(v)}(v)) \leq b_v$.

\begin{lemma}\label{lem:app9}
For arbitrary graph $G$ and arbitrary demands, we can find a feasible demand matching $\solI \sse \sol$ with $p(\solI) \geq p(\sol) \cdot \iratio$.
\end{lemma}
\begin{proof}
For each vertex $v$, label $v$ randomly with \texttt{s} with probability $\alpha$ or with $\texttt{l}$ with  probability $1-\alpha$ (for $\alpha$ to be chosen later).
Say $e \in \sol$ {\em agrees} with the labelling for an endpoint $v$ if either $e \in S(v)$ and $v$ is labelled \texttt{s}, or $v \in L(v)$ and $v$ is labelled \texttt{l}.
Let $\sola \subseteq \sol$ be the edges agreeing with the labelling on both endpoints.

Modify the graph $(V, \sola)$ by replacing each $v \in V$ labelled \texttt{s} with $|\delta_{\sola}(v)|$ vertices and reassigning the endpoint $v$ of each $e \in \delta_{\sola}(v)$ to one of these vertices in a one-to-one fashion.
See Figure \ref{fig:splitting} for an illustration. Call this new graph $\overline G$.
\begin{figure}
\begin{center}
\includegraphics[width=0.6\textwidth]{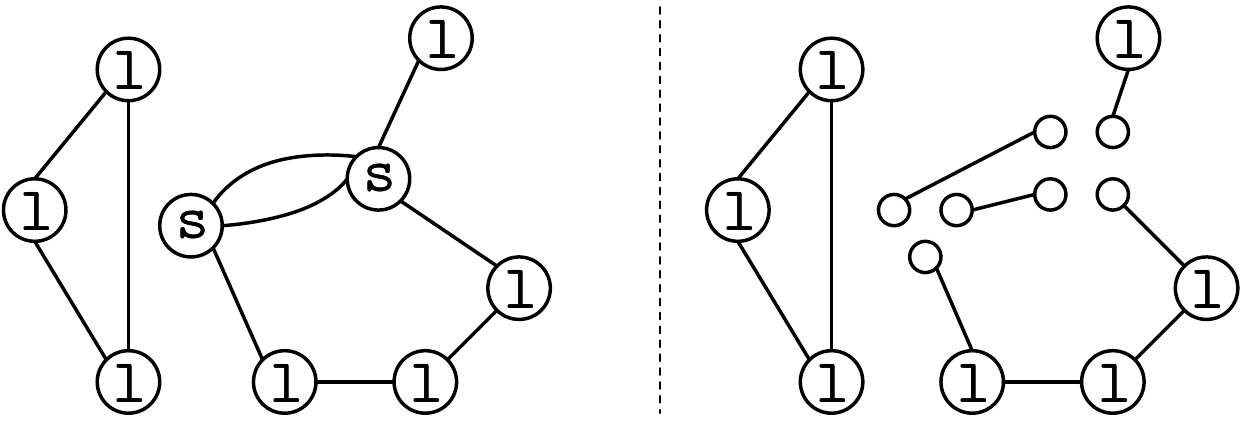}
\end{center}
\caption{
{\bf Left}: The graph with vertex labels \texttt{s} and \texttt{l} and edges $\sola$.
{\bf Right}: The graph $\overline G$ obtained by ``shattering'' the \texttt{s} vertices. Notice the maximum degree is 2, the \texttt{ss} edges are isolated, and the \texttt{sl} edges lie on paths.
}
\label{fig:splitting}
\end{figure}

Each vertex in $\overline G$ has degree at most 2 so $\overline G$ decomposes naturally into paths and cycles. Each path with $\geq 2$ edges can be decomposed into 2 matchings and each cycle can be
decomposed into 3 matchings.
Randomly choose one such matching for each path and cycle to keep and discarding the remaining edges on these paths and cycles. Note edges $uv$ of $\overline G$ where $u$ and $v$ both had degree 1 are not discarded.

Let $\solI$ be the resulting set of edges, viewed in the original graph $G$. Note that $\solI$ is feasible: any vertex labelled \texttt{s} already had its capacity satisfied by $\sola$ because $\delta_{\sola}(v) \subseteq S(v)$.
Any vertex labelled $\texttt{l}$ has at most one of its incident edges in $\sola$ chosen to stay in $\solI$.

Let $e = uv \in \sol$, we place a lower bound on ${\bf Pr}[e \in \solI]$ by analyzing a few cases.
\begin{itemize}
\item If $e \in S(u) \cap S(v)$, then ${\bf Pr}[e \in \solI] = {\bf Pr}[e \in \sola] = \alpha^2$.
\item If $e \in S(u) \cap L(v)$ or vice-versa then ${\bf Pr}[e \in \solI] = \alpha \cdot (1-\alpha)/2$ (note $e$ does not lie on a cycle in $\overline G$ since one endpoint is labelled \texttt{s}).
\item If $e \in L(u) \cap L(v)$ then ${\bf Pr}[e \in \solI] = (1-\alpha)^2/3$.
\end{itemize}
Choosing $\alpha = 2/5$, we have ${\bf E}[p(\solI)] \geq p(\sol) \cdot \frac{3}{25}$.
\end{proof}
We can efficiently derandomize this technique as follows. First, we use a pairwise independent family of random values to generate a probability space over labelings of $V$ with $O(|V|)$ events such that the
distribution of labels over pairs $u,v \in V$ is the same as with independently labelling the vertices. See Chapter 11 of \cite{RANDBOOK} for details of this technique.
For each such labelling, we decompose the paths and cycles of
$\overline G$ into matchings and keep the most profitable matching from each path and cycle instead of randomly picking one.

\begin{lemma}
For a conflict-free instance of \gdmm, we can find a feasible solution $\solI \sse \sol$ with $p(\solI) \geq \frac{p(\sol)}{4}$.
\end{lemma}

\begin{proof}
The set $\sola$ from the proof of Lemma \ref{lem:app9} is already feasible so it does not need to be pruned further. In this case, choose $\alpha = 1/2$.
\end{proof}

\begin{lemma}
If the given graph $G$ is bipartite, then we can find a feasible solution $\solI \sse \sol$ with $p(\solI) \geq \frac{p(\sol)}{7}$. 
\end{lemma}

\begin{proof}
Say $V_L, V_R$ are the two sides of $V$. We first partition $\sol$ into 4 groups:
$$ \{uv \in \sol : uv \in S(u) \cap S(v)\} ~~~~~~\{uv \in \sol : uv \in L(u) \cap L(v)\} $$
$$ \{uv \in \sol : uv \in S(u) \cap L(v)\} ~~~~~~\{uv \in \sol : uv \in L(u) \cap S(v)\} $$
The first set is feasible. The latter three sets can each be partitioned into two feasible sets as follows. For one of these sets, form $\overline G$ as in the proof of Lemma \ref{lem:app9}. Each cycle can also be decomposed
into two matchings because $G$, thus $\overline G$, is bipartite. Between all sets listed above, we have partitioned $\sol$ into 7 feasible sets. Let $\solI$ be one with maximum profit.
\end{proof}

\begin{lemma}
For an arbitrary graph $G = (V,E)$ with a consistent ordering on edges, we can find a feasible demand matching $\solI \sse \sol$ with $p(\solI) \geq \frac{p(\sol)}{5}$.
\end{lemma}

\begin{proof}
We partition $\sol$ into five groups in this case. Consider the edges in decreasing order of the consistent ordering. When edge $e = uv$ is considered, assign it to a group that does not include edges in $L(u) \cup L(v)$
that come before $e$ in the ordering. As $|L(u) \cup L(v)| \leq 4$, the edges can be partitioned into five groups this way.
Each group $\sola$ is a feasible demand matching since $\delta_\sola(v) \subseteq S(v)$ or $|\delta_\sola(v)| = 1$ for each vertex $v$.
Now let $\solI$ be the group with maximum profit, so $p(\solI) \geq \frac{p(\sol)}{5}$.
\end{proof}


\begin{lemma}\label{lem:dmm_small}
If $d_{v,e} \leq \epsilon \cdot b_v$ for each $v \in V, e \in \delta(v)$, we can find a feasible demand matching $\solI \subseteq \sol$ with $p(\solI) \geq (1-O(\epsilon^{1/3})) \cdot p(\sol)$.
\end{lemma}
This is proven using a common randomized pruning procedure.
See, for example, \cite{CCKR11} for a similar treatment in another packing problem.
\begin{proof}
As the bound is asymptotic, we assume $\epsilon$ is sufficiently small for the bounds below to hold. 
Recall that under the assumption of the lemma that $\sol$ satisfies $p(\sol) \geq \OPT\eqref{lp-m}$ and $d_v(\delta_{\sol}(v)) \leq b_v + 2 \cdot \max_{e \in \delta(v)} d_{v,e} \leq (1+2\epsilon) \cdot b_v$.

Let $\delta = \epsilon^{1/3}$.
We initially let $\sola$ be a subset of $\sol$ by independently adding each $e \in \sol$ to $\sola$ with probability $1-\delta$. We then prune $\sola$ to a feasible set $\solI$ as follows.

Process each the edges of $\sola$ in arbitrary order. When considering a particular $e \in \sola$, add $e$ to $\solI$ only if $\solI \cup \{e\}$ is feasible. We have
$${\bf Pr}[e \not\in \solI] = {\bf Pr}[e \not\in \sola] + {\bf Pr}[e \not\in \solI | e \in \sola] \cdot {\bf Pr}[e \in \sola] \leq \delta + {\bf Pr}[e \not\in \solI | e \in \sola],$$
we proceed to bound the last term.
For each endpoint $v$ of $e$, 
consider the random variable ${\bf D}^e_v = d_v(\delta_\sola(v) - \{e\})$. If $e \not\in \solI$ yet $e \in \sola$, then an endpoint $v$ of $e$ has ${\bf D}^e_v > (1-\epsilon) \cdot b_v$.
As this event is independent of $e \in \sola$, it suffices to bound $\Pr[{\bf D}^e_v > (1-\epsilon)\cdot b_v]$.

Note $$\mu_v := {\bf E}[{\bf D}^e_v] = (1-\delta)\cdot d_v(\delta_{\sol}(v)-\{e\}) \leq (1-\delta) \cdot (1+2\epsilon) \cdot b_v \leq (1-\delta/2) \cdot b_v.$$
For brevity, let $\sigma_v := {\bf Var}[{\bf D}^e_v]$ and for $e' \in \sol$ let ${\bf X}_{e'}$ be the random variable indicating $e' \in \sola$. As the edges are added to $\sola$ independently,
\begin{eqnarray*}
\sigma_v & = & \sum_{e' \in \delta_{\sol}(v)-\{e\}} {\bf Var}[d_{v,e'} {\bf X}_{e'}] \leq \sum_{e' \in \delta_{\sol}(v)-\{e\}} {\bf E}[(d_{v,e'} {\bf X}_{e'})^2]  \\
 & \leq & \sum_{e' \in \delta_{\sol}(v)-\{e\}} d_{v,e'}^2 \leq \sum_{e' \in \delta_{\sol}(v)-\{e\}} d_{v,e'} \cdot (\epsilon \cdot b_v) \leq\epsilon(1+2\epsilon) \cdot b_v^2 \leq 2\epsilon \cdot b_v^2
\end{eqnarray*}
Chebyshev's inequality states ${\bf Pr}(|{\bf D}^e_v - \mu_v| \geq a) \leq \frac{\sigma_v}{a^2}$ for any $a > 0$. Using $a = \delta \cdot (1-2\epsilon)  b_v$
and, for sufficiently small $\epsilon$, the fact that
$$\mu_v + \frac{\delta}{3} \cdot b_v \leq (1-\delta/6)\cdot b_v \leq (1-\epsilon) \cdot b_v$$
we see
\begin{eqnarray*}
{\bf Pr}[{\bf D}^e_v > (1-\epsilon)\cdot b_v] & \leq & {\bf Pr}[|{\bf D}^e_v - \mu_v| > \frac{\delta}{3}\cdot b_v] \\
 & \leq & \frac{3^2}{\delta^2}\cdot\frac{2\epsilon \cdot b_v^2}{b_v^2} = \frac{18\epsilon}{\delta^2}.
\end{eqnarray*}

Finally, using the union bound over both endpoints of $e$, we have ${\bf Pr}[e \not\in \solI] \leq \delta + \frac{36\epsilon}{\delta^2} = 37 \epsilon^3$.
That is, ${\bf E}[p(\solI)] \geq (1-O(\epsilon^{1/3})) \cdot p(\sol).$
\end{proof}
This analysis only uses the second moment method, so it can derandomize efficiently using a pairwise-independent family of random variables. Again, see Chapter 11 of \cite{RANDBOOK} for a discussion of this technique.



\section{Strong NP-Hardness of \dm for Simple, Bipartite Planar Graphs}\label{app:hard}
\begin{proof}[Proof of Theorem \ref{thm:planar_hard}]
For an instance $\Phi = (X,C)$ of \textsc{SAT} with variables $X$ and clauses $C$, let $G_\Phi$ denote the graph with vertices $X \cup C$ and edges connecting $x \in X$ to $c \in C$ if $x$ appears in $c$
(either positively or negatively).
Consider the restriction of \textsc{SAT} to instances $\Phi$ where $G_\phi$ is planar and can be drawn such that for every clause $x \in X$, the vertices for clauses $c$ that contain the positive literal $x$
appear consecutively around $x$ (thus, so do the vertices for clauses containing the negative literal $\overline{x}$).
Such instances were proven to be NP-hard in \cite{L82}. The reduction in \cite{L82} reduces from an arbitrary planar SAT instance and it is clear from the reduction
that if we reduce from bounded-degree planar SAT, then the resulting SAT instance also has bounded degree.

So, let $\Phi = (X,C)$ be an instance of planar sat where each vertex in $G_\Phi$ has degree at most some universal constant $D$ and for each $x \in X$ the edges connecting $x$ to clauses $c$ that contain
the positive literal $x$ appear consecutively around $x$.

Our \dm instance has vertices $\{u_c : c \in C\} \cup \{t_x, f_x, v_x : x \in X\}$ and the following edges.
\begin{itemize}
\item for each $x \in X$, two edges $v_xt_x$ and $v_xf_x$, both with demand and profit $D$.
\item for each $x \in X$, a unit demand/profit edge $u_ct_x$ for every $c \in C$ including $x$ negatively.
\item for each $x \in X$, a unit demand/profit edge $u_cf_x$ for every $c \in C$ including $x$ positively.
\end{itemize}
Each vertex $u_c$ for clauses $c \in C$ has capacity 1 and $t_x, f_x, v_x$ all have capacity $D$ for all $x \in X$. See Figure \ref{fig:reduction} for an illustration. Note the resulting graph is bipartite,
with $\{t_x, f_x : x \in X\}$ forming one side of the bipartition.

We claim $\Phi$ is satisfiable if and only if the optimum \dm solution has value $D \cdot |X| + |C|$.
Suppose $\Phi$ is satisfiable. For each $x \in X$, select edge $t_xv_x$ if $x$ is \texttt{true} in the satisfying assignment, otherwise select $f_xv_x$. As each $c \in C$ is satisfied, some literal in $c$ is satisfied.
Select the corresponding incident edge.

Conversely, consider an optimal demand matching solution $F$. Without loss of generality, we may assume $|F \cap \{t_xv_x, f_xv_x\}| = 1$ for each $x \in X$. Indeed, because of the capacity of $v_x$ we cannot choose both.
If neither is chosen, then $F' = (F - \delta(t_x)) \cup \{v_xt_x\}$ is also feasible and has no smaller value. The value of $F$ is then $D \cdot |X|$ plus the number of edges of $F$ incident to some $u_c, c \in C$.
Consider the truth assignment that assigns $x$ \texttt{true} if $t_xv_x \in F$ and \texttt{false} if $f_xv_x \in F$. Then some edge incident to $c$ can be in $F$ if and only if this truth assignment satisfies $c$.
\end{proof}

\begin{figure}
\begin{center}
\includegraphics[width=0.8\textwidth]{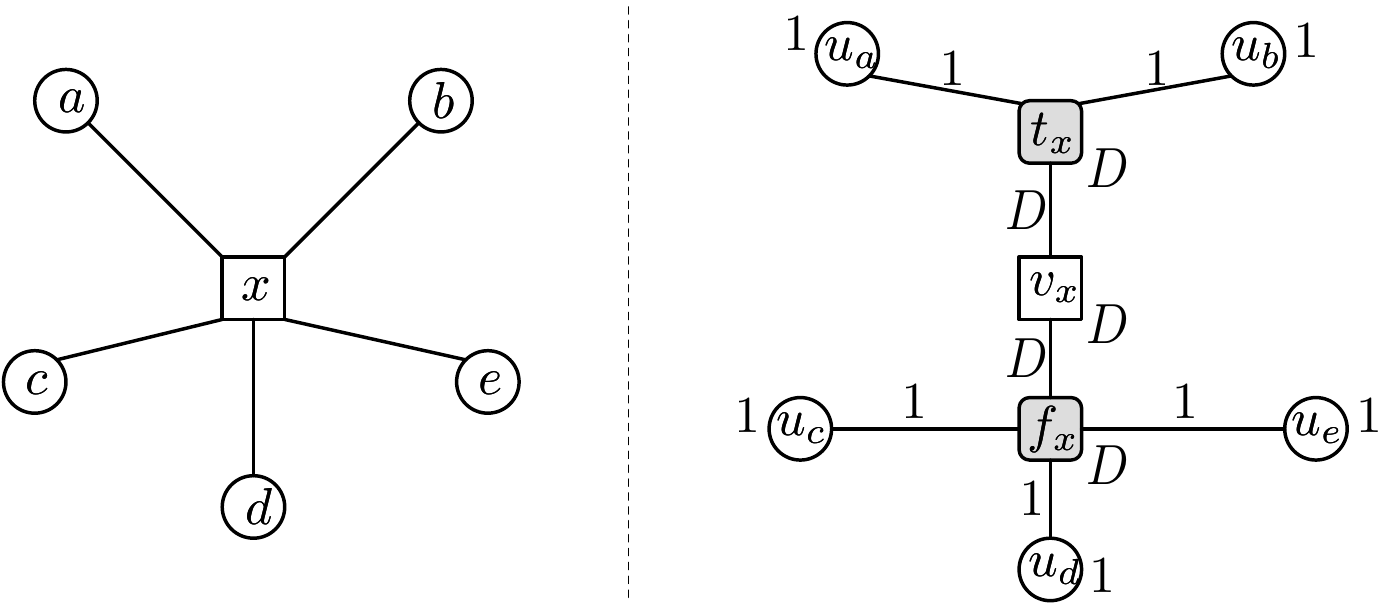}
\end{center}
\caption{
{\bf Left}: Part of the graph $G_\Phi$ for an instance of planar \textsc{SAT}. Here, $x$ is a variable that appears {\em negatively} in clauses $a,b$ and {\em positively} in clauses $c,d,e$.
In particular, not that all positive occurrences of $x$ appear consecutively around the vertex for $x$.
{\bf Right}: The corresponding vertices in the \dm instance. The numbers indicate the vertex capacities and the edge demands and profits. Here $d$ is the maximum degree of a vertex variable
in $G_{\Phi}$, which may be regarded as a constant. The shading of the vertices $t_x, f_x$ illustrates that the graph in the \dm instance is bipartite.
}
\label{fig:reduction}
\end{figure}

\section{Demand Matching in Excluded-Minor Families}\label{sec:minor}

In this section we prove \gdms admits a PTAS in graphs that exclude a fixed graph as a minor. Our proof of Theorem \ref{thm:planar_hard} (the {\bf NP}-hardness) appears in the full version.
Throughout we let $\OPT$ denote the optimum solution value to the given \gdms instance.

%
%
%



Let $H$ be a graph and let $\mathcal G_H$ be all graphs that exclude $H$ as a minor. Our PTAS uses the following decomposition.
\begin{theorem}[Demaine, Hajiaghayi, and Kawarabayashi \cite{EHK05}]\label{thm:exclude}
There is a constant $c_H$ depending only on $H$ such that for any $k$ and any $G \in \mathcal G_H$, the vertices $V$ of $G$ can be partitioned into $k+1$ disjoint sets so that the union of any $k$ of these sets
induce a graph with treewidth bounded by $c_H \cdot k$. Such a partition can be found in time that is polynomial in $|V|$.
\end{theorem}
Using this decomposition in a standard way, we get a PTAS for \gdms when $G \in \mathcal G_H$ if we have a PTAS for \gdms in bounded-treewidth graphs.

Algorithm \ref{alg:ptas} summarizes the steps to reduce \gdms from bounded-genus graphs to bounded-treewidth graphs. We also note that the tree decomposition itself can be executed in polynomial
time (e.g. \cite{B96}) since $c_H \cdot k$ is regarded as a constant.
\begin{algorithm}
\caption{High-Level algorithm for the \gdms PTAS for graphs excluding $H$ as a minor} \label{alg:ptas}
\begin{algorithmic}
\State $k \gets 3/\epsilon$
\State let $\pi \gets \{V_1, V_2, \ldots, V_{k+1}\}$ be a partition of $V$ as in Theorem \ref{thm:exclude}
\For{each $V_i \in V$}
\State compute a tree decomposition of $G[V-V_i]$ with treewidth $c_H \cdot k$
\State use a $(1-\epsilon/3)$-approximation on $G[V-V_i]$ to get a demand matching solution $\solI_i$
\EndFor
\State \Return the best solution $\solI_i$ found
\end{algorithmic}
\end{algorithm}

\begin{theorem}
For any constant $\epsilon > 0$ and any fixed graph $H$, Algorithm \ref{alg:ptas} runs in polynomial time and is a $(1-\epsilon)$-approximation.
\end{theorem}
\begin{proof}
That the algorithm runs in polynomial time is clear given that $H$ and $\epsilon$ are regarded as constant and the fact that the PTAS for bounded-treewidth graphs runs in polynomial time for constant treewidth and $\epsilon$.

Each edge $e$ is excluded from $G[V-V_i]$ for at most 2 parts $V_i$ of the partition $\pi$. Thus, the optimum \gdms solution for some $G[V-V_i]$ has profit at least $(1-2/(k+1)) \cdot \OPT \geq (1- 2\epsilon/3) \cdot \OPT$.
So the returned solution has value at least $(1-\epsilon/3) \cdot (1-2\epsilon/3) \cdot \OPT \geq (1-\epsilon) \cdot \OPT$.
\end{proof}

Intuition for our approach is given at the end of Section \ref{sec:results}.
We assume, for simplicity, that all $d_{v,e}$-values are distinct so we can naturally speak of {\em the} largest demands in a set. This is without loss of generality, we could scale demands and capacities by a common value so they
are integers and then subtract $\frac{2i + j}{3|E|^2}$ from the $j$'th endpoint of the $i$'th edge according to some arbitrary ordering. Such a perturbation does not change feasibility of solutions
as the total amount subtracted from all edges is $< 1$.


\subsection{A Sparsification Lemma}

We present our sparsification lemma, which even holds for general instances of \gdmm. 
\begin{lemma}[Sparsification Lemma]\label{lem:sparse}
For each $\epsilon > 0$ there is a feasible solution $\solI \subseteq E$ with the following properties.
\begin{itemize}
\item $p(\solI) \geq (1-2\epsilon) \cdot \OPT$
\item for each $v \in V$, there is some $M_v \subseteq \solI$ with $|M_v| \leq 1/\epsilon^2$ such that 
$d_{v,e} \leq \epsilon \cdot (b_v - d_v(\delta_{M_v}(v)))$ for all $e \in \delta_{\solI-M_v}(v)$
\end{itemize}
\end{lemma}

Think of $M_v$ as the ``large'' edges in $\delta_\solI(v)$ and $\delta_{\solI-M_v}(v)$ as the ``small'' edges in $\delta_\solI(v)$. Note that some $e\in \solI$ may be designated large on one endpoint and small on the other.
\begin{proof}
Let $\solI^*$ be an optimum solution. For each $v \in V$, if $|\delta_{\solI^*}(v)| \geq 1/\epsilon^2$  then 
let $L_v$ be the $1/\epsilon^2$ edges in $\delta_{\solI^*}(v)$ with greatest $d_v$-demand and $R_v$ be a random subset of $L_v$ of size $1/\epsilon$.
If $|\delta_{\solI^*}(v)| < 1/\epsilon^2$, simply let $L_v = \delta_{\solI^*}(v)$ and $R_v = \emptyset$.


Set $\solI = \solI^* - \cup_{v \in V}R_v$ and for each $v \in V$ set $M_v = \solI \cap L_v$.
Clearly $\solI$ is feasible as it is a subset of the optimum solution.
For each $e = uv \in \solI^*$, $e$ lies in $R_u$ or $R_v$ with probability at most $\epsilon$ each, so ${\bf Pr}[e \not\in \solI] \leq 2\epsilon$. Thus, ${\bf E}[p(\solI)] \geq (1-2\epsilon) \cdot \OPT$.

Now we focus on proving the second property for $\solI$. Let $v$ be an arbitrary vertex in $V$. By construction $|M_v| \leq |L_v| \leq 1/\epsilon^2$. If $|R_v| = 0$ then $\delta_{\solI-M_v}(v) = \emptyset$, otherwise, $|R_v| = 1/\epsilon$ and for each remaining $e \in \delta_{\solI-M_v}(v)$, we note that
$d_{v,e} + d_v(\delta_{M_v}(v)) + \sum_{e' \in R_v} d_{v,e'} \leq b_v$
because the terms represent a subset of edges of $\solI^*$ incident to $v$.
Rearranging and using the fact that $d_{v,e'} \geq d_{v,e}$ for any $e' \in R_v$ shows $\frac{1}{\epsilon} \cdot d_{v,e} \leq b_v - d_v(\delta_{M_v}(v))$.
\end{proof}

This motivates the following notion of a relaxed solution.
\begin{definition}
An $\epsilon$-relaxed solution is a subset $\solI \subseteq E$ along with sets $M_v \subseteq \delta_\solI(v)$ with $|M_v| \leq 1/\epsilon^2$ for each $v \in V$ such that the following hold.
First, let $\overline b_v = b_v - d_v(\delta_{M_v}(v))$ for each $v \in V$.
Next, for each $e \in \delta_{\solI-M_v}(v)$, let $d'_{v,e}$ be the value of $d_{v,e}$ rounded down to the nearest integer multiple of $\frac{\epsilon}{|E|} \overline b_v$.
Then the following must hold:
\begin{itemize}
\item {\bf Large Edge Feasibility}: $d_v(\delta_{M_v}(v)) \leq b_v$ for each $v \in V$.
\item {\bf Small Edges}: $d_{v,e} \leq \epsilon \overline b_v$ for each $v \in V$ and each $e \in \delta_{\solI-M_v}(v)$.
\item {\bf Discretized Small Edge Feasibility}: $d'_v(\delta_{\solI-M_v}(v)) \leq \overline{b_v}$ for each $v \in V$
\end{itemize}
\end{definition}
The set $\solI$ in an $\epsilon$-relaxed solution is not necessarily a feasible \gdms solution under the original demands $d$. As we will see shortly, it can be pruned to get a feasible solution without losing much value.
Note the scaling from $d$ to $d'$ for some of the edges $e$ in the definition is done independently for each endpoint of $e$: the demand at different endpoints may be shifted down by different amounts.

Sometimes we informally say just a set $\solI \subseteq E$ itself is an $\epsilon$-relaxed solution even if we do not explicitly mention the corresponding $M_v$ sets.
\begin{lemma}
Let $\solI$ be an $\epsilon$-relaxed solution with maximum possible value $p(\solI)$. Then $p(\solI) \geq (1-2\epsilon) \cdot OPT$.
\end{lemma}
\begin{proof}
The set $\solI$ and its corresponding $M_v$ subsets from Lemma \ref{lem:sparse} suffice.
\end{proof}

\begin{lemma}\label{lem:prune_excl}
Given any $\epsilon$-relaxed solution $\solI \subseteq E$,
we can efficiently find some $\sol \subseteq \solI$ that is a feasible \gdms solution with $p(\sol) \geq (1-O(\epsilon^{1/3})) \cdot p(\solI)$.
\end{lemma}
The idea is that the $\{0,1\}$ indicator vector of $\solI$
is almost a feasible solution to \eqref{lp-m} with the trivial matroid $\mathcal I = 2^E$ in the residual instance after all ``large'' edges are packed so it can be pruned to a feasible solution while losing very little value
by appealing to the last bound in Theorem \ref{thm:main}. There is a minor subtlety in how to deal with edges that are both ``small'' and ``large''.
\begin{proof}
Let $\overline b_v = b_v - d_v(\delta_{M_v}(v))$.
Consider the following modified instance of \gdms. The graph is $G = (V,\solI)$, each $v \in V$ has capacity $\overline b_v$, and the demands are
$$\widehat{d}_{v,e} = \left\{
\begin{array}{rl}
d_{v,e} & e \not\in M_v \\
0 & e \in M_v.
\end{array}
\right.$$
Note some edges may have one of their endpoint's demands set to 0 while the other is unchanged. For each $v \in V$.
$$\widehat{d}_v(\delta_{\solI}(v)) = d_v(\delta_{\solI-M_v}(v)) \leq d'_v(\delta_{\solI-M_v}(v)) + |\delta_{\solI-M_v}(v)| \cdot \frac{\epsilon}{|E|} \cdot \overline b_v \leq (1+\epsilon) \cdot \overline b_v.$$
Therefore, setting $x_e = \frac{1}{1+\epsilon}$ yields a feasible solution for \eqref{lp-m} (with the trivial matroid in which all subsets are independent)
with value $\frac{p(\solI)}{1+\epsilon}$.
By Theorem \ref{thm:main}, we can efficiently find a feasible \gdms solution $\sol$ such that
$$p(\sol) \geq  (1-O(\epsilon^{1/3}) \frac{p(\solI)}{(1+\epsilon)} \geq (1-O(\epsilon^{1/3}) p(\solI).$$
\end{proof}
Alternatively, we could avoid solving an LP and simply prune $\solI$ using a similar approach as in the proof of Lemma \ref{lem:dmm_small}.


\subsection{A Dynamic Programming Algorithm}
Suppose $G = (V,E)$ has treewidth at most $\tau$ and that we are given a tree decomposition $\mathcal T = (\mathcal B, E_{\mathcal T})$ of $G$ where each $B \in \mathcal B$ has $|B| \leq \tau+1$.
Recall this means the following:
\begin{enumerate}
\item For each $v \in V$, the set of bags $\mathcal B_v = \{B \in \mathcal B : v \in B\}$ form a connected subtree of $\mathcal T$.
\item For each $uv \in E$, there is at least one bag $B \in \mathcal B$ with $u,v \in B$.
\end{enumerate}
Let $B^r \in \mathcal B$
be some arbitrarily chosen {\em root} bag. View $\mathcal T$ as being rooted at $B^r$. We may assume that each $B \in \mathcal B$ has at most two children. In fact, it simplifies our recurrence a bit to assume
each $B \in \mathcal B$ is either a leaf in $\mathcal T$ or has precisely two children. This is without loss of generality.
Arbitrarily order the children of a non-leaf vertex so one is the {\em left} child and one is the {\em right} child.
For a bag $B$, let $\mathcal T_B$ be the subtree of $\mathcal T$ rooted at $B$ (so $\mathcal T_{B^r} = \mathcal T$).

For each $v \in V$, say $\overline{B}_v$ is the bag containing $v$ that is closest to the root $B^r$.
Note for $uv \in E$ with $\overline{B}_u \neq \overline{B}_v$ that one of $\overline B_u$ or $\overline B_v$ lies on the path between the other and $B^r$ (by the properties of tree decompositions).
For each $B \in \mathcal B$ and each $v \in B$, we partition a subset of the edges of $\delta(v)$ into four groups:
\begin{itemize}
\item $\delta^{\there}(v:B) = \{uv \in \delta(v) : \overline B_u = B\}$.
\item $\delta^{\tleft}(v:B) = \{uv \in \delta(v) : \overline B_u \text{ lies in the left subtree of } B \}$.
\item $\delta^{\tright}(v:B) = \{uv \in \delta(v) : \overline B_u \text{ lies in the right subtree of } B \}$.
\item $\delta^{\tup}(v:B) = \{uv \in \delta(v) : \overline B_u \text{ lies between } B \text{ and } B^r \}$.
\end{itemize}
The only other edges $uv \in \delta(v)$ not accounted for here do not have $\overline B_u$ in either $\mathcal T_B$ or between $B$ and $B^r$.
We note if $B = \overline B_v$, then every edge in $\delta(v)$ lies in one of the four groups and for any $uv \in \delta^{\tup}(v:B)$ we must have $u \in B$ (otherwise no bag contains $u$ and $v$, which is impossible since $uv \in E$)
and, consequently, $\overline B_u$ lies between $B$ and $B^r$. This will be helpful to remember when we describe the recurrence.

~

\noindent
{\bf Dynamic Programming States}\\
Let $\Delta := \{\there, \tleft, \tright, \tup\}$ be the set of ``directions'' used above.
The DP states are given by tuples $\Phi$ with the following components.
\begin{itemize}
\item A bag $B \in \mathcal B$.
\item For each $v \in B$, a subset $M_v \subseteq \delta(v)$ with $|M_v| \leq 1/\epsilon^2$.
\item For each $v \in B$ and $\kappa \in \Delta$, an integer $a_{v,\kappa} \in \{0, \ldots, |E|/\epsilon\}$ such that $\sum_{\kappa \in \Delta} a_{v, \kappa} \leq \frac{|E|}{\epsilon}$.
\end{itemize}
The number of such tuples is at most $|\mathcal B| \cdot |E|^{O(\tau/\epsilon^2)}\cdot (|E|/\epsilon)^{O(\tau)}$, which is polynomial in $G$ when $\tau$ and $\epsilon$ are regarded as constants.
The idea behind $a_{v,\kappa}$ is that it describes how to reserve the discretized $d'_v$-demand for edges $uv \in \delta^\kappa(v:B)- M_v$.
Of course, other edges in $\delta(v)$ not in the partitions
$\delta^\kappa(v:B)$ may be in an optimal $\epsilon$-relaxed solution. They will either be explicitly guessed in $M_v$ or will be considered in a state higher up the tree by the time the bag $\overline B_v$ is processed.

~

\noindent
{\bf Dynamic Programming Values}\\
For each such tuple $\Phi = \left( B; \langle M_v\rangle_{v \in B}; \langle a_{v,\kappa}\rangle_{v \in B, \kappa \in \Delta}\right)$, we let $f(\Phi)$
denote the maximum total value of an $\epsilon$-relaxed solution $\sol \subseteq E$ (with corresponding large sets $M'_v$ for $v \in V$) satisfying the following properties.
We slightly abuse notation and say $v \in \mathcal T_B$ for some $v \in V$ if $v$ lies in some bag of the subtree $\mathcal T_B$.
\begin{itemize}
\item Each $uv \in \sol$ has at least one endpoint in $\mathcal T_B$.
\item $M'_v = M_v$ for each $v \in B$.
\item Each $uv \in \sol$ with both $\overline B_u, \overline B_v \not\in \mathcal T_B$ lies in $M'_u \cup M'_v$.
\item For $v \in B$ let $\overline{b_v} = b_v - d_v(\delta_{M'_v}(v))$.
For $\kappa \in \Delta$ and $v \in B$, it must be that $d'_v(\delta^{\kappa}(v:B) \cap \sol - M'_v) \leq  a_{v,\kappa} \cdot \frac{\epsilon}{|E|} \cdot \overline{b_v}$
where $d'_{v,e}$ is the largest integer multiple of $\frac{\epsilon}{|E|} \cdot \overline{b_v}$ that is at most $d_{v,e}$ for $e \in \delta_{\sol-M'_v}(v)$.
\end{itemize}
The last point is a bit technical. Intuitively, it says the scaled demand of small edges incident to $v$ coming from some direction $\kappa \in \Delta$ fit in the capacity of $v$ reserved for that direction.


%
If there is no such $F$, we say $f(\Phi) = -\infty$. Note the maximum of $f(\Phi)$ over all configurations $\Phi$ for the root bag $B_r$ is the maximum value over all $\epsilon$-relaxed
solutions.

\subsubsection{The Recurrence: Overview}
We start by outlining the main ideas. A tuple $\Phi$ is a base case if the bag $B$ is a leaf
of $\mathcal T$. In this case, only edges in some $\delta^{\kappa}(v:B)$ set for $\kappa \in \{\there, \tup\}$ are considered (there are none in the directions \tleft, \tright). We find the optimal
way to pack such edges that are not part of a ``large'' set $M_v$ while ensuring the $d'_v$-demands do not violate the residual capacities $\overline{b_v}$ and, in particular, for each direction
$\kappa$ we ensure this packing does not violate the part of the residual capacity for that direction allocated by the $a_{v, \kappa}$ values. This subproblem is just the \textsc{Multi-Dimensional Knapsack}
problem with $2|B|$ knapsacks. A standard pseudopolynomial-time algorithm can be used to solve it as the scaled demands are from a polynomial-size discrete range.

For the recursive step, we try all pairs of configurations $\Phi^\tleft, \Phi^\tright$ that are ``consistent'' with $\Phi$. Really this just means they agree on the sets $M_v$ for shared vertices $v$ and
they agree on how much demand $a_{v,\kappa}$ should be allocated for each direction. For each such consistent pair, we pack small edges in $\delta^\there(v:B)$ and $\delta^\tup(v:B)$
optimally such that their scaled demands do not violate the $a_{v,\kappa}$-capacities, again using \textsc{Multi-Dimensional Knapsack}.

\subsubsection{The Recurrence: Details}
Fix a tuple $\Phi = \left( B; \langle M_v\rangle_{v \in B}; \langle a_{v,\kappa}\rangle_{v \in B, \kappa \in \Delta}\right)$. We describe how to compute $f(\Phi)$ recursively.
In the recursive step, all subproblems invoked will involve only children of $B$ and base cases are leaves of $\mathcal T$. So we can use dynamic programming to compute
$f(\Phi)$ in polynomial time; it will be evident that evaluating the cases in terms of subproblems take polynomial time.

For brevity, let $M_{\texttt{big}} = \cup_{v \in B} M_v$.
As with the discussion above, for this tuple we let $\overline{b_v} = b_v - d_v(\delta_{M_v}(v))$.
Let $E_\Phi$ be all edges $e$ such that:
\begin{itemize}
\item $e \in \delta^{\kappa}(v:B)-M_{\texttt{big}}$ for some endpoint $v$ of $e$ lying in $B$ and some $\kappa \in \{\tup, \there\}$,
\item for any such endpoint $v \in B$ and associated $\kappa \in \{\tup, \there\}$ we have $d_{v,e} \leq \epsilon \overline{b_v}$, and
\item at least one endpoint $v$ has $\overline{B}_v = B$.
\end{itemize}

Some of the $M_v$ edges in the configuration may be small on the other endpoint which also lying in $B$. With this in mind, for $v \in B$ and $\kappa \in \{\tup, \there\}$
we let $D'^\kappa_v = d'_v(\delta^\kappa(v:B) \cap M_\texttt{big} - M_v)$ be the scaled demand (from the appropriate direction) on $v$ from edges guessed explicitly by $\Phi$ yet are small on $v$.
Think of this as the small demand across $v$ that we are required to pack due to the guesses for large edges, the rest of the calculation for $f(\Phi)$ will be to optimally pack the edges
in $E_{\Phi}$ into the remaining allocated capacities.

In both the base case and recursive step, we require the following to hold or else we set $f(\Phi) = \infty$.
First, $M_{\texttt{big}}$ is feasible by itself (i.e. $\overline b_v \geq 0$ for each $v \in B$).
Next, $D'^\kappa_v \leq a_{v,\kappa} \cdot \frac{\epsilon}{|E|} \cdot \overline b_v$ for each $v \in B, \kappa \in \{\tup, \there\}$.
This means the edges already guessed in $\Phi$ that are small on $v \in B$ have their scaled demands fit in the $d'_v$-values.

~

\noindent
{\bf Base Case}\\
Suppose $B$ is a leaf of $\mathcal T$.
In this case, no demand comes from edges contributing to $a_{v, \tleft}$ or $a_{v, \tright}$
and we can use dynamic programming to find the maximum-value set of edges of $E_\Phi$ to pack in the capacities $a_{v,\kappa}$ for $v \in B, \kappa \in \{\there, \tup\}$.



That is, we find a maximum-profit $F \subseteq E_\Phi$ such that for each $v \in B, \kappa \in \{\there, \tup\}$,
$$d'_v(\delta^\kappa(v:B) \cap F) \leq a_{v,\kappa}\cdot \frac{\epsilon}{|E|} \cdot \overline{b_v} - D'^{\there}_v.$$
This can be done using a standard pseudo-polynomial time dynamic programming algorithm for \textsc{Multiple-Dimensional Knapsack} with $2|B|$ knapsacks (one for the $\tup$ entry and one for the $\there$ entry of each vertex).
Note each entry of the table is indexed by an integer multiple of $\frac{\epsilon}{|E|} \cdot \overline b_v$, so this runs in polynomial time.

~

\noindent
{\bf Recursive Step}\\
The idea behind computing $f(\Phi)$ when $\Phi$ is not a base case is to try all pairs of configurations $\Phi^\tleft, \Phi^\tright$ for the children of the bag $B$ for $\Phi$ that agree with
$\Phi$ on the ``large'' edges $B_v$ for shared vertices $v$ and on how the $a_{v,\kappa}$ scaled capacity allocations for $v$ are distributed. For any such pair, we use a similar dynamic programming
routine as in the base case to pack in the maximum value of ``small'' edges that contribute to $a_{v,\kappa}$ for $\kappa \in \{\there, \tup\}$ and $v \in B$. We now make this precise.

Let $B$ be a non-root bag with children $B^{\tleft}$ and $B^{\tright}$. Say tuples $\Phi^\tleft$ for $B^\tleft$ and $\Phi^\tright$ for $B^\tright$ are {\em compatible} with $\Phi$ if the following hold.
We use notation like $B^\tleft$ to indicate the bag-component of $\Phi^\tleft$, $a^\tright_{v,\kappa}$ for the $a_{v,\kappa}$-component of $\Phi^\tright$, etc.
\begin{itemize}
\item The sets $M_v, M^\tleft_v, M^\tright_v$ are all the same (when they exist) for $v \in B$.
\item Let $a^\tleft_v = \sum_{\kappa \in \Delta-\{\tup\}} a^\tleft_{v,\kappa}$ be the total $d'$-demand of small edges contributing to $v$'s load whose other endpoint $u$ has $\overline{B}_u$ in $\mathcal T_{B^\tleft}$.
Similarly define $a^\tright_v$ for the right child of $B$. Then for each $v \in B$ it must be $a^\tleft_v = a_{v, \tleft}$ (if $v \in B^\tleft$) and $a^\tright_v = a_{v, \tright}$ (if $v \in B^\tright$).
In other words, $\Phi$ agrees with $\Phi^{\tleft}$ and $\Phi^{\tright}$ about how much $d'$-demand is used $\mathcal T_{B^\tleft}$ and $\mathcal T_{B^\tright}$ for each $v \in B$.
\item For each $v \in B$, if $v \in B^\tleft$ then $a^\tleft_{v, \tup} = a_{v, \there} + a_{v, \tup}$ and if $v \in B^\tright$ then $a^\tright_{v, \tup} = a_{v, \there} + a_{v, \tup}$.
\end{itemize}

Let $g(\Phi)$ be the maximum value of a subset of $E_\Phi$ such that the remaining $a_{v,\kappa}$-capacity of $v \in B$ for each $\kappa \in \{\there, \tup\}$ are not overpacked by the $d'$-values of these edges.
This is essentially identical to the calculation in the base case using \textsc{Multi-Dimensional Knapsack}.

Finally we compute
$$
f(\Phi) = g(\Phi) + \displaystyle \max_{\substack{\Phi^\tleft, \Phi^\tright\\ \text{ compatible with } \Phi}} \bigg[f(\Phi^\tleft) + f(\Phi^\tright) + p(\cup_{v \in B-(B^\tleft\cup B^\tright)} M_v) - p(\cup_{v \in B^\tleft \cap B^\tright} M_v)\bigg].
$$
This adds the value obtained from the two subproblems, subtracts out the ``double-counted part'' which is exactly the set of ``big'' edges from both subproblems, adds the new big edges for $\Phi$,
and also the new small edges that were packed by the inner DP algorithm.

We feel one comment is in order to see why no other edges are double counted. Consider an edge $uv$ where $uv$ contributes to both solutions of some pair $\Phi^\tleft, \Phi^\tright$.
If some subtree, say $\mathcal T_{B^\tleft}$, does not contain either $\overline B_u$ or $\overline B_v$ then $uv \in M_u \cup M_v$ by definition of $f(\Phi)$ and such edges were subtracted in the expression 
above to avoid double counting.
If some subtree contains both $\overline B_u$ and $\overline B_v$ then $uv$ could not have contributed to the subproblem. So some subtree contains $\overline B_u$ and the other contains $\overline B_v$.
But then no bag contains both $u$ and $v$, which is impossible.

\bibliographystyle{plain}

\bibliography{dm-ref}


\appendix

\end{document}